\documentclass[11pt]{article}

\usepackage[top=3cm, bottom=2cm, left=3.5cm, right=3.5cm, includefoot]{geometry}
\usepackage{authblk}
\usepackage{amsthm}
\usepackage{url}
\usepackage{microtype}
\newtheorem{theorem}{Theorem}
\newtheorem{lemma}{Lemma}
\newtheorem{definition}{Definition}

\usepackage{amssymb}
\usepackage{amsmath}
\usepackage{graphicx}
\usepackage{multirow}
\usepackage{colortbl}
\usepackage[rightcaption]{sidecap}

\usepackage[algoruled,linesnumbered,algo2e,lined,noend]{algorithm2e}

\newcommand{\etal}{\textit{et~al.}}
\newcommand{\para}[1]{{#1}}
\newcommand{\set}[1]{{\bf{#1}}}

\newcommand{\shta}[1]{\mathit{SHT}({#1})}
\newcommand{\ph}[1]{\mathit{PH}({#1})}

\newcommand{\psufp}[1]{\mathit{psp}({#1})}
\newcommand{\pre}[1]{\mathit{prev}({#1})}

\newcommand{\pph}[1]{\mathit{PPH}({#1})}
\newcommand{\apph}[1]{\mathit{APPH}({#1})} 
\newcommand{\occ}{\mathit{occ}}
\newcommand{\rt}{\mathit{root}}
\newcommand{\node}[1]{{#1}}
\newcommand{\mrp}[1]{\mathit{mrp}(#1)}
\newcommand{\pmrp}[1]{\mathit{pmrp}(#1)}
\newcommand{\descendant}[2]{\mathit{Des}_{#1}({#2})}
\newcommand{\zero}{\set{Z}}

\newcommand{\calN}{{\cal N}}

\newcommand{\primary}[1]{\mathit{prim}(#1)}
\newcommand{\secondary}[1]{\mathit{sec}(#1)}

\newcommand{\newactn}{\mathit{newnode}}
\newcommand{\ans}{\mathit{ans}}

\newcommand{\child}[2]{\mathit{child}(#1,#2)}
\newcommand{\depth}[1]{\mathit{depth}(#1)}
\newcommand{\norm}[2]{\mathit{normalize}(#1,#2)}
\newcommand{\nll}{\textsf{null}}
\newcommand{\create}{\textsf{create}}

\newcommand{\ignore}[1]{}

\title{Position Heaps for Parameterized Strings}

\author[1]{Diptarama}
\author[1]{Takashi Katsura}
\author[1]{Yuhei Otomo}
\author[1]{Kazuyuki Narisawa}
\author[1]{Ayumi Shinohara}
\affil[1]{Graduate School of Information Sciences, Tohoku University, \\
  6-6-05 Aramaki Aza Aoba, Aoba-ku, Sendai, Japan\\
  \texttt{\{diptarama@shino., katsura@shino., otomo@shino., narisawa@, ayumi@\}ecei.tohoku.ac.jp}}

\begin{document}
\maketitle

\begin{abstract}
We propose a new indexing structure for parameterized strings, called parameterized position heap.
Parameterized position heap is applicable for parameterized pattern matching problem,
 where the pattern matches a substring of the text if there exists a
 bijective mapping from the symbols of the pattern to the symbols of the substring.
We propose an online construction algorithm of parameterized position heap of a text and show that our algorithm runs in linear time
with respect to the text size.
We also show that by using parameterized position heap, we can find all occurrences of a pattern in the text in linear time with respect to the product of the pattern size and the alphabet size.

\end{abstract}

\section{Introduction}
String matching problem is to find occurrences of a pattern string in a text string.
Formally, given a text string $t$ and a pattern string $p$ over an
alphabet $\Sigma$, output all positions at which $p$ occurs in $t$.
Suffix tree and suffix array are most widely used data structures
and provide many applications for various string matchings~(see e.g.~\cite{GUS,JEWELS}).

Ehrenfeucht~\etal{}~\cite{PH} proposed an indexing structure for string matching,
called a {\em position heap}.
Position heap uses less memory than suffix tree does,
and provides efficient search of patterns by preprocessing the text string,
similarly to suffix tree and suffix array.
A position heap for a string $t$ is a {\em sequence hash tree}~\cite{SHT}
for the ordered set of all suffixes of $t$.
In \cite{PH}, the suffixes are ordered in the ascending order of length,
and the proposed construction algorithm processes the text from right to left.
Later, Kucherov~\cite{OPH} considered the ordered set of suffixes
in the descending order of length
and proposed a linear-time {\em online} construction algorithm
based on the Ukkonen's algorithm~\cite{UKKONEN}.
Nakashima {\em et al.}~\cite{PosHeapTrie} proposed an algorithm
to construct a position heap for a set of strings,
where the input is given as a {\em trie} of the set.
Gagie~\etal{}~\cite{Gagie2013} proposed a position heap with limited height
and showed some relations between position heap and suffix array.

The parameterized pattern matching that focuses on a structure of strings
is introduced by Baker~\cite{PST}.
Let $\Sigma$ and $\Pi$ be two disjoint sets of symbols.
A string over $\Sigma \cup \Pi$ is called a {\em parameterized string}
(p-string for short).
In the parameterized pattern matching problem,
given p-strings $\para{t}$ and $\para{p}$,
find positions of substrings of $\para{t}$
that can be transformed into $\para{p}$ by applying one-to-one function
that renames symbols in $\Pi$.
The parameterized pattern matching is motivated
by applying to the software maintenance~\cite{PM,PST,PMA},
the plagiarism detection~\cite{EPM},
the analysis of gene structure~\cite{SST},
and so on.
Similar to the basic string matching problem,
some indexing structures that support the parameterized pattern matching
are proposed,
such as parameterized suffix tree~\cite{PST}, structural suffix tree~\cite{SST}, and parameterized suffix array~\cite{PSC2008-8,I2009}.

In this paper,
we propose a new indexing structure called \emph{parameterized position heap}
for the parameterized pattern matching.
The parameterized position heap is
a sequence hash tree for the ordered set
of prev-encoded~\cite{PST} suffixes of a parameterized string.
We give an online construction algorithm of a parameterized position heap based on Kucherov's algorithm~\cite{OPH}
that runs in $O(n\log{(|\Sigma|+|\Pi|)})$ time
and an algorithm that runs in $O(m\log{(|\Sigma|+|\Pi|)} + m|\Pi| + \occ)$ time
to find the occurrences of a pattern in the text,
where $n$ is the length of the text, $m$ is the length of the pattern, $|\Sigma|$ is the number of constant symbols, 
 $|\Sigma|$ is the number of parameter symbols, and $\occ$ is the number of occurrences of the pattern in the text.

\section{Notation}

Let $\Sigma$ and $\Pi$ be two disjoint sets of symbols.
$\Sigma$ is a set of \emph{constant} symbols
and $\Pi$ is a set of \emph{parameter} symbols.
An element of $\Sigma^{*}$ is called a \emph{string},
and an element of $(\Sigma \cup \Pi)^{*}$ is called a \emph{parameterized string},
or \emph{p-string} for short.
For a p-string $\para{w}=\para{xyz}$, $\para{x}$, $\para{y}$, and $\para{z}$ are called 
\emph{prefix}, \emph{substring}, and \emph{suffix} of $\para{w}$, respectively.
$|\para{w}|$ denotes the length of $\para{w}$, and
$\para{w}[i]$ denotes the $i$-th symbol of $\para{w}$ for $1 \leq i \leq |\para{w}|$.
The substring of $\para{w}$ that begins at position $i$ and ends at position $j$ 
is denoted by $\para{w}[i:j]$ for $1 \leq i \leq j \leq |\para{w}|$.
Moreover,
let $\para{w}[:i]=\para{w}[1:i]$ and $\para{w}[i:]=\para{w}[i:|w|]$ for $1 \leq i \leq |\para{w}|$.
The empty p-string is denoted by $\varepsilon$, that is $|\varepsilon|=0$.
For convenience, let $\para{w}[i:j] = \varepsilon$ if $i > j$.
Let $\calN$ denote the set of all non-negative integers.

Given two p-strings $w_1$ and $w_2$, 
$w_1$ and $w_2$ are a \emph{parameterized match} or \emph{p-match}, denoted by $w_1 \approx w_2$,
if there exists a bijection $f$ from the symbols of $w_1$ to the symbols of $w_2$,
such that $f$ is identity on the constant symbols~\cite{PST}.
We can determine whether $w_1 \approx w_2$ or not by using an encoding called \emph{prev-encoding} defined as follows.

\begin{definition}[Prev-encoding~\cite{PST}]
For a p-string $w$ over $\Sigma \cup \Pi$,
the \emph{prev-encoding} for $w$,
 denoted by  $\pre{w}$,
is a string $x$ of length $|w|$ over $\Sigma \cup \calN$ defined by
\[
  x[i] = \begin{cases}
    \para{w}[i] & \mbox{if }\para{w}[i] \in \Sigma , \\
    0           & \mbox{if }\para{w}[i] \in \Pi \mbox{ and } \para{w}[i] \neq \para{w}[j] \mbox{ for } 1 \le j < i,\\

    i-\max\{j \mid \para{w}[j]=\para{w}[i] \mbox{ and } 1 \le j < i\} & \mbox{otherwise} .
  \end{cases}
\]
\end{definition}

For any p-strings $w_1$ and $w_2$,  $w_1 \approx w_2$ if and only if $\pre{w_1}=\pre{w_2}$.
For example, given $\Sigma = \{{\tt a, b}\}$ and $\Pi = \{u,v,x,y\}$,
$s_1 = uvuv{\tt a}uuv{\tt b}$ and $s_2 = xyxy{\tt a}xxy{\tt b}$ are p-matches where
$\pre{w_1} = \pre{w_2} = 0022{\tt a}314{\tt b}$.

The parameterized pattern matching is a problem to find occurrences of a p-string pattern in a p-string text defined as follows.
\begin{definition}[Parameterized pattern matching~\cite{PST}]
Given two p-strings, text $t$ and pattern $p$,
find all positions $i$ in $t$ such that $t[i:i+|p|-1]\approx p$. 
\end{definition}
For example, let us consider a text $t=uv{\tt a}u{\tt b}u{\tt a}v{\tt b}v$ and a pattern $p = x{\tt a}y{\tt b}y$
over $\Sigma = \{{\tt a, b}\}$ and $\Pi = \{u,v,x,y\}$.
Because $p \approx t[2:6]$ and $p \approx t[6:10]$, we should output $2$ and $6$.

Throughout this paper, let $t$ be a text of length $n$ and $p$ be a pattern of length $m$.

\section{Position Heap}

\begin{figure}[t]
	\centering
	\begin{minipage}[t]{0.32\hsize}
		\centering
		\includegraphics[scale=0.54]{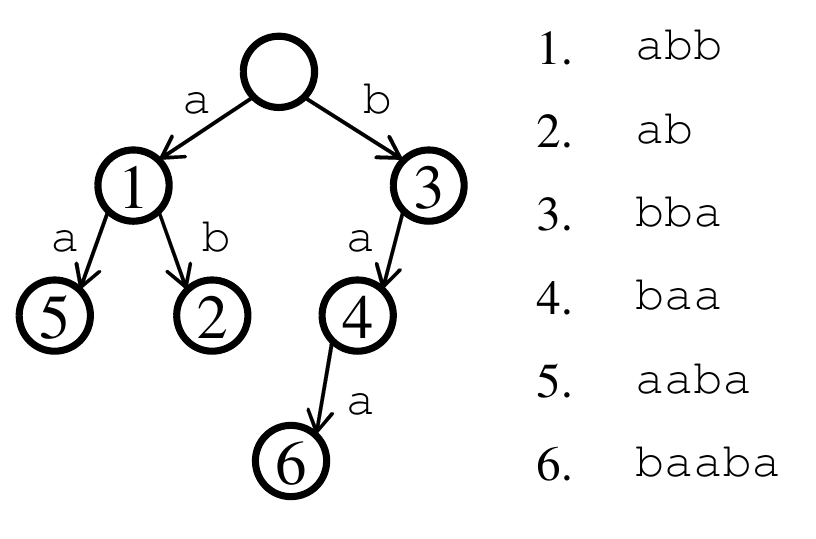}\\
		\ \ \ \scriptsize{(a)}
	\end{minipage}
	\begin{minipage}[t]{0.32\hsize}
		\centering
		\includegraphics[scale=0.4]{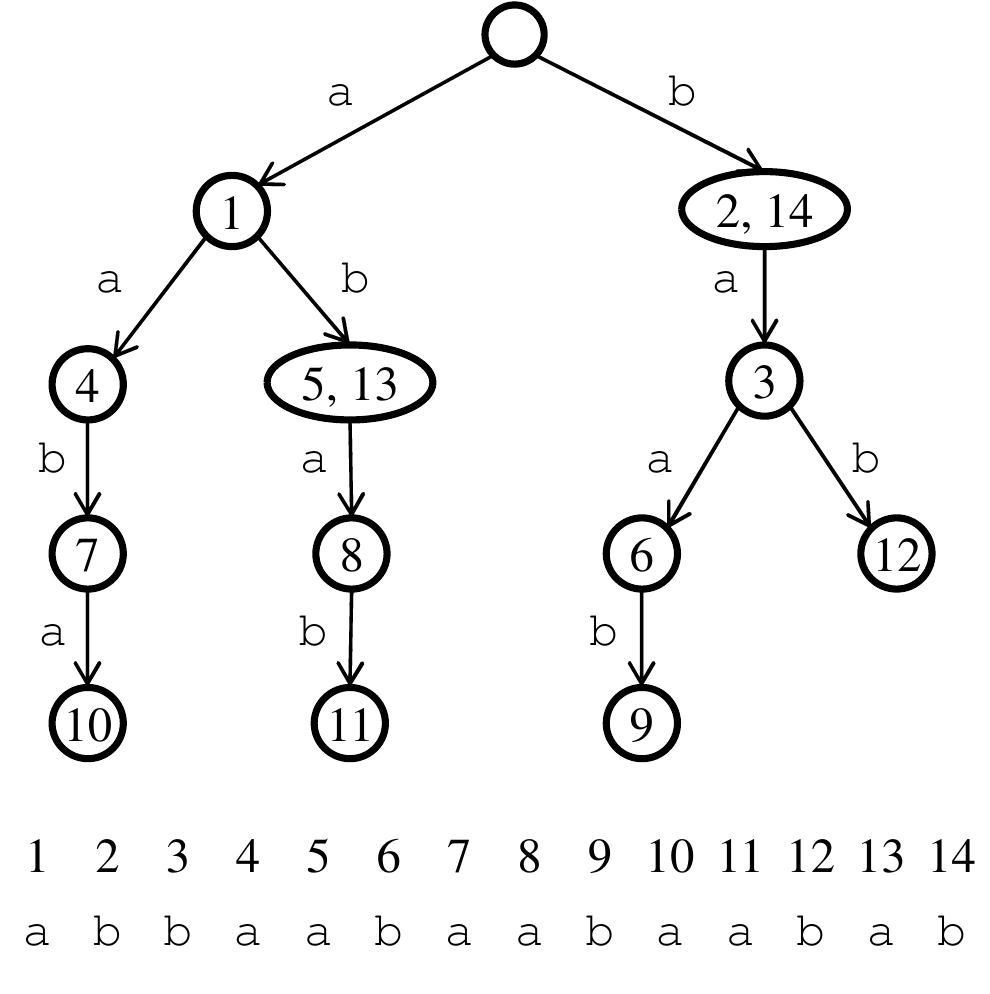}\\
		\ \ \ \scriptsize{(b)}
	\end{minipage}
	\begin{minipage}[t]{0.32\hsize}
		\centering
		\includegraphics[scale=0.4]{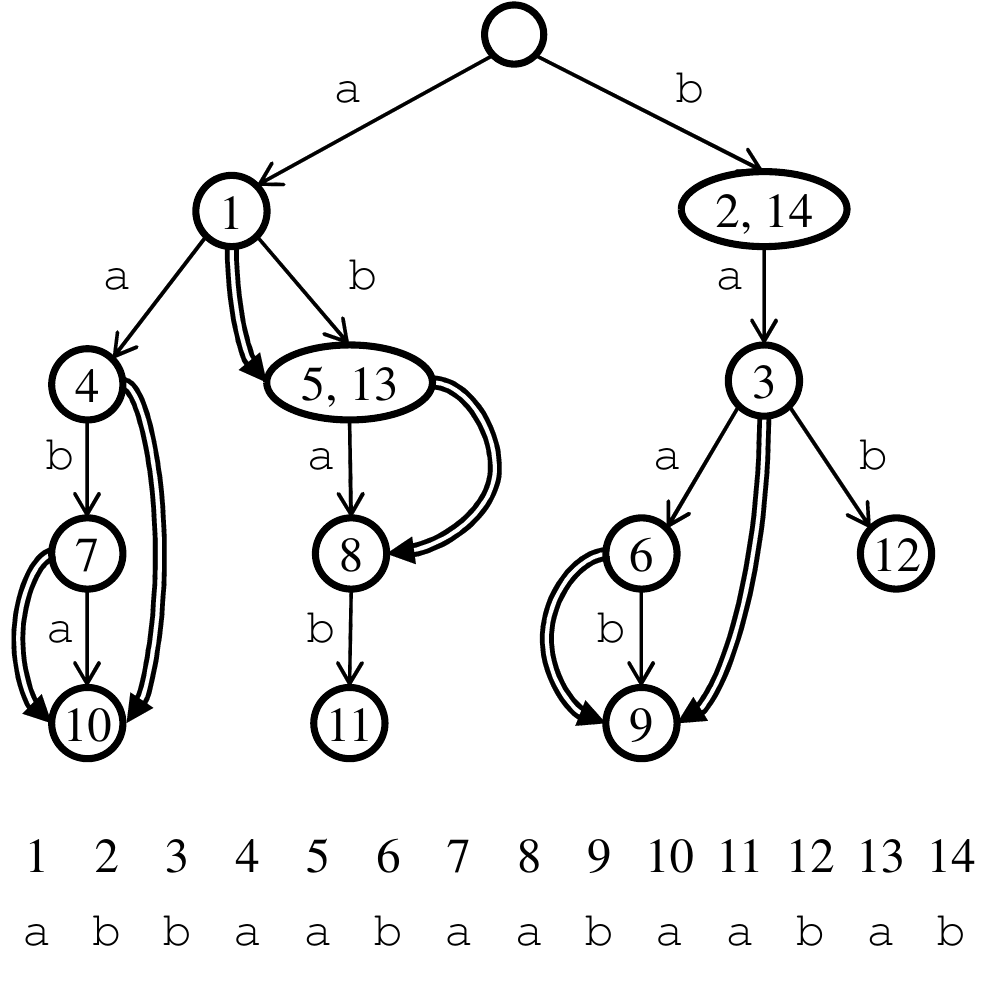}\\
		\ \ \ \scriptsize{(c)}
	\end{minipage}
	\caption{
		(a) A sequence hash tree for $({\tt aab,ab,bba,baa,aaba,baaba})$.
		(b) A position heap for a string ${\tt abbaabaabaabab}$,
		(c) An augmented position heap for a string ${\tt abbaabaabaabab}$.
		Maximal-reach pointers for $\mrp{i} \neq i$ are illustrated by doublet arrows.
	}
	\label{fig:sht_ph}
\end{figure}

In this section, we briefly review the position heap for strings.
First we introduce the \emph{sequence hash tree}
that is a trie for hashing proposed by Coffman and Eve~\cite{SHT}.
Each edge of the trie is labeled by a symbol
and each node can be identified with the string obtained by concatenating all labels found on the path from root to the node.

\begin{definition}[Sequence Hash Tree]
Let $\set{W}=(w_1,\ldots,w_n)$ be an ordered set of strings over $\Sigma$ and $\set{W}_i=(w_1,\ldots,w_i)$ for $1 \leq i \leq n$.
A \emph{sequence hash tree} $\shta{\set{W}}=(V_n,E_n)$
for $\set{W}$ is a trie over $\Sigma$ defined recursively as follows.
Let $\shta{\set{W}_i} = (V_i, E_i)$. Then,
	\begin{align*}
	\shta{\set{W}_i} &=  
		\begin{cases}
		(\{\varepsilon\}, \emptyset)	&	\left( \text{if $i = 0$} \right),\\
		(V_{i-1} \cup \{\node{p_i}\}, E_{i-1} \cup \{(\node{q_i},c,\node{p_i})\})
			& \left( \text{if $1 \leq i \leq n$} \right).
		\end{cases}
	\end{align*}
where $p_i$ is the shortest prefix of $w_i$ such that $\node{p_i} \not\in V_{i-1}$,
and $q_i=w_i[1:|p_i|\!-\!1]$, $c=w_i[|p_i|]$.
If no such $p_i$ exists,
then $V_i = V_{i-1}$ and $E_i = E_{i-1}$. 
\end{definition}

Each node in a sequence hash tree stores one or several indices of strings in the input set.
An example of a sequence hash tree is shown in Fig.~\ref{fig:sht_ph}~(a).

The \emph{position heap} proposed by Ehrenfeucht~\etal{}~\cite{PH}
is a sequence hash tree for the ordered set of all suffixes of a string.
Two types of position heap are known.
The first one is proposed by Ehrenfeucht~\etal{}~\cite{PH},
that constructed by the ordered set of suffixes
in ascending order of length
and the second one is proposed by Kucherov~\cite{OPH}, 
which constructed in descending order.
We adopt the Kucherov~\cite{OPH} type and his online construction algorithm
for constructing position heaps for parameterized strings in Section~\ref{sec:p-poisition heap}.
Here we recall the definition of the position heap by Kucherov.
\begin{definition}[Position Heap~\cite{OPH}]
Given a string $t \in \Sigma^n$,
let $\set{S}_t = (t[1:], t[2:], \dots, t[n:])$ be the ordered set of all suffixes of $t$ except $\varepsilon$
in descending order of length.
The \emph{position heap} $\ph{t}$ for $t$ is $\shta{\set{S}_t}$.
\end{definition}

Each node except the $\rt$ in a position heap stores either one or two integers
those are beginning positions of corresponding suffixes. 
We call them \emph{regular node} and \emph{double node} respectively.
Assume that $i$ and $j$ are positions stored
by a double node $v$ in $\ph{t}$ where $i < j$,
$i$ and $j$ are called the \emph{primary position} and the \emph{secondary position}
respectively.
Fig.~\ref{fig:sht_ph}~(b) shows an example of a position heap. 

In order to find occurrences of the pattern in $O(m+\occ)$ time, Ehrenfeucht~\etal{}~\cite{PH} and Kucherov~\cite{OPH}
added additional pointer called \emph{maximal-reach pointer} to the position heap
and called this extended data structure as \emph{augmented position heap}.
An example of an augmented position heap is showed in Fig.~\ref{fig:sht_ph}~(c). 

\section{Parameterized Position Heap} \label{sec:p-poisition heap}
In this section,
we propose a new indexing structure called
\emph{parameterized position heap}.
It is based on the position heap proposed by Kucherov~\cite{OPH}.

\subsection{Definition and Property of Parameterized Position Heap}

\begin{figure}[t]
	\centering
	\begin{minipage}[t]{0.6\hsize}
		\centering
		\includegraphics[scale=0.4]{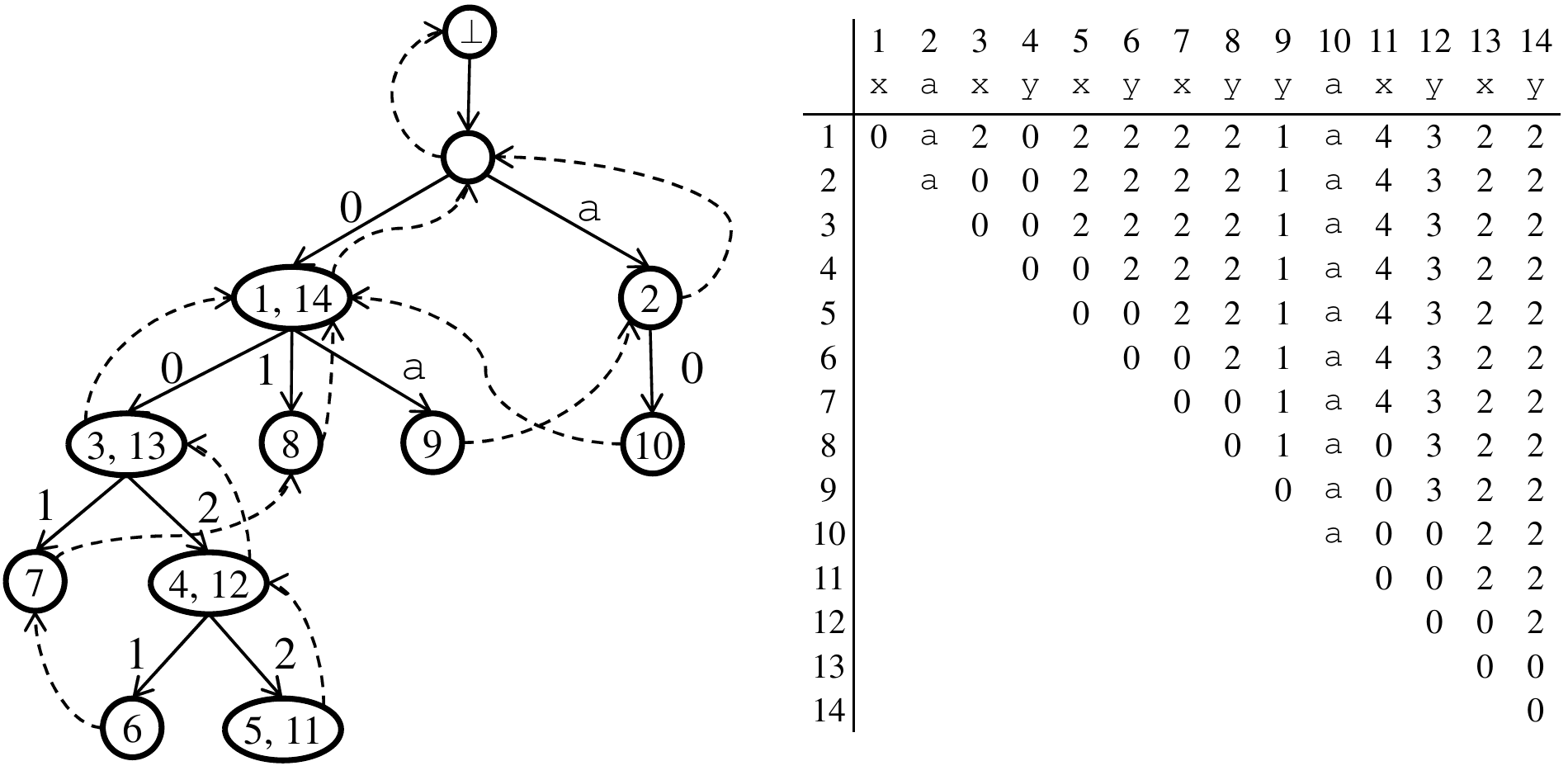}\\
		\ \ \ \scriptsize{(a)}
	\end{minipage}
	\begin{minipage}[t]{0.39\hsize}
		\centering
		\includegraphics[scale=0.4]{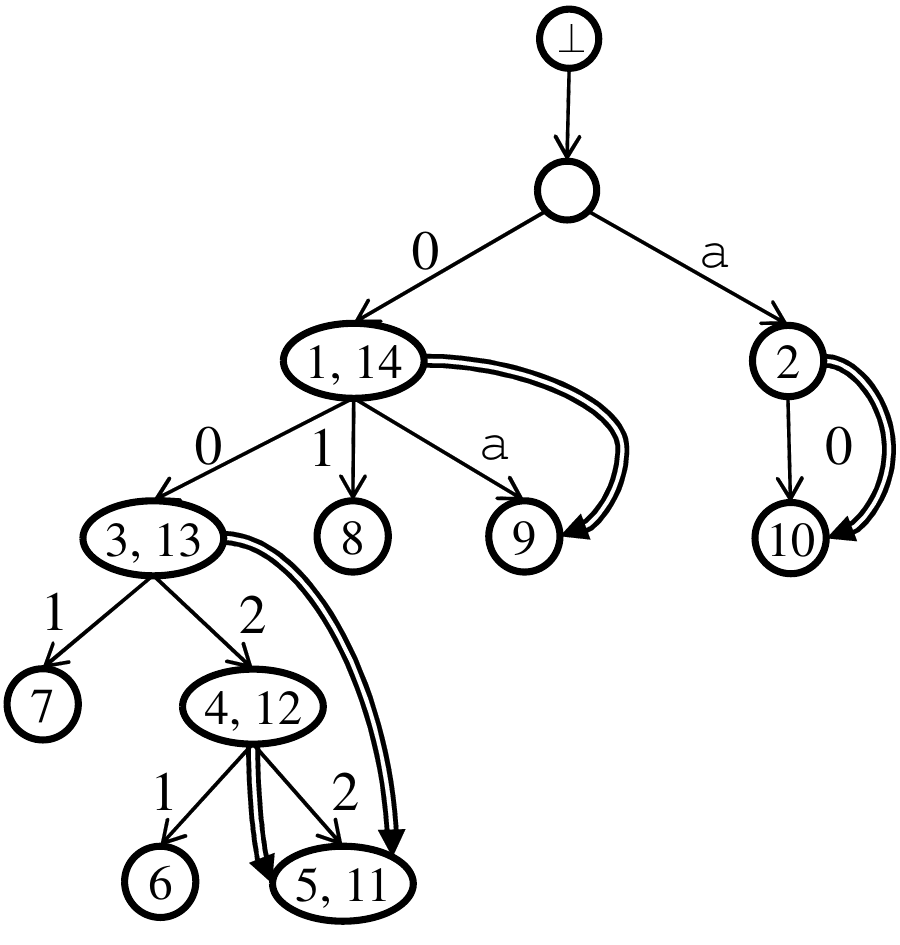}\\
		\ \ \ \scriptsize{(b)}
	\end{minipage}
	\caption{
		Let $\Sigma=\{ {\tt a}\}$, $\Pi=\{x,y\}$ and $\para{t}=x{\tt a}xyxyxyy{\tt a}xyx$.
		(a) A parameterized position heap $\pph{\para{t}}$.
		Broken arrows denote suffix pointers.
		(b) An augmented parameterized position heap $\apph{\para{t}}$.
		Parameterized maximal-reach pointers for $\pmrp{i} \neq i$ are illustrated by doublet arrows.
	}
	\label{fig:pph}
\end{figure}

The parameterized position heap
is a sequence hash tree~\cite{SHT} for the ordered set of
prev-encoded suffixes in the descending order of length.


\begin{definition}[Parameterized Position Heap]
Given a p-string $\para{t} \in (\Sigma \cup \Pi)^n$,
let $\set{\para{S}}_\para{t} = (\pre{\para{t}[1:]}, \pre{\para{t}[2:]}, \dots, \pre{\para{t}[n:]})$ be the ordered set of all prev-encoded suffixes
of the p-string $\para{t}$ except $\varepsilon$ in descending order of length.
The \emph{parameterized position heap} $\pph{\para{t}}$ for $\para{t}$ is $\shta{\set{\set{\para{S}}}_\para{t}}$.
\end{definition}


Fig.~\ref{fig:pph} (a) shows an example of a parameterized position heap. 
A parameterized position heap $\pph{\para{t}}$
for a p-string $\para{t}$ of length $n$ consists of 
the root and nodes that corresponds to $\pre{\para{t}[1:]}, \pre{\para{t}[2:]}, \dots, \pre{\para{t}[n:]}$,
so $\pph{\para{t}}$ has at most $n+1$ nodes.
Each node in $\pph{\para{t}}$ holds either one or two of beginning positions of corresponding p-suffixes
similar to the standard position heaps.
We can specify each node in $\pph{\para{t}}$ by its primary position, its secondary position,
or the string obtained by concatenating labels found on the path from the root to the node.

Different from standard position heap,
$\pre{\para{t}[i:]} = \pre{\para{t}}[i:]$ does not necessarily hold for some cases.
For example, for $\para{t}={\tt xaxyxyxyyaxyxy}$, $\pre{\para{t}[3:]}=0022221{\tt a}4322$ while
$\pre{\para{t}}[3:]=0222221{\tt a}4322$.
Therefore, the construction and matching algorithms for 
the standard position heaps cannot be directly applied for the parameterized position heaps.
However, we can similar properties to construct parameterized position heaps efficiently.

\begin{lemma}
\label{lem:pph_substr}
For $i$ and $j$, where $1 \leq i \leq j \leq n$,
if $\node{\pre{\para{t}[i:j]}}$ is represented in $\pph{\para{t}}$,
then a prev-encoded string for any substring of $\para{t}[i:j]$
is also represented in $\pph{\para{t}}$.
\end{lemma}
\begin{proof}
	First we will show that prev-encoding of any prefix of $t[i:j]$ is represented in $\pph{t}$.
	From the definition of prev-encoding,
	$\pre{\para{t}[i:j]}[1:i-j]=\pre{\para{t}[i:j-1]}$.
	In other words, $\pre{\para{t}[i:j-1]}$ is a prefix of $\pre{\para{t}[i:j]}$.
	From the definition of $\pph{\para{t}}$,
	prefixes of $\pre{\para{t}[i:j]}$ are represented in $\pph{t}$.
	Therefore, $\pre{\para{t}[i:j-1]}$ is represented in $\pph{\para{t}}$.
	Similarly, $\pre{\para{t}[i:j-2]}$, $\cdots$, $\pre{\para{t}[i:i]}$
	are represented in $\pph{\para{t}}$.
	
	Next, we will show that prev-encoding of any suffix of $t[i:j]$ is represented in $\pph{t}$.
	From the above discussion, there are positions $b_0 < b_1 < \cdots < b_{j-i} = i$ in $t$ such that
	$\pre{t[b_k:b_k+k]} = \pre{t[i:i+k]}$.
	From the definition of parameterized position heap, $\pre{t[b_{1}+1:b_{1}+{1}]}$ is represented in $\pph{t}$.
	Since $\pre{t[b_{k}+1:b_{k}+{k}]}$ is a prefix of $\pre{t[b_{k+1}+1:b_{k+1}+{k+1}]}$ for $0 < k < j-i$,
	if $\pre{t[b_{k}+1:b_{k}+{k}]}$ is represented in $\pph{t}$ then $\pre{t[b_{k+1}+1:b_{k+1}+{k+1}]}$ is also represented in $\pph{t}$ recursively.
	Therefore, $\pre{t[b_{j-i}+1:b_{j-i}+{j-i}]}= \pre{t[i+1:j]}$ is represented in $\pph{t}$.
	Similarly, $\pre{\para{t}[i+2:j]}$, $\cdots$, $\pre{\para{t}[j:j]}$ are represented in $\pph{\para{t}}$.
	
	Since any prefix and suffix of $\pre{t[i:j]}$ is represented in $\pph{t}$,
	we can say that any substring of $\pre{t[i:j]}$ is represented in $\pph{t}$ by induction.
\end{proof}


\subsection{Online Construction Algorithm of Parameterized Position Heap}
In this section,
we propose an online algorithm that constructs parameterized position heaps.
Our algorithm is based on Kucherov's algorithm, although it cannot be applied easily.
The algorithm updates $\ph{t[1:k]}$ to $\ph{t[1:k+1]}$
when $t[k+1]$ is read, where $1 \leq k \leq n-1$.
Updating of the position heap begins from a special node,
called the {\em active node}.
A position specified by the active node is called the {\em active position}.
At first, we show that there exists a position similar to the active position
in the parameterized position heap.

\begin{lemma}
\label{lem:pph_active}
If $j$ is a secondary position of a double node in
a parameterized position heap, then $j+1$ is also a secondary position.
\end{lemma}
\begin{proof}
	Let $i$ be the primary position and $j$ be the secondary position of node $v$,
	where $i < j$.
	This means there is a position $h$ such that $\pre{t[i:h]} = \pre{t[j:]}$.
	By Lemma~\ref{lem:pph_substr},
	there is a node that represents $\pre{t[i+1:h]}$.
	Since $\pre{t[j+1:]} = \pre{t[i+1:h]}$, then $j+1$ will be the secondary positions of node $\pre{t[i+1:h]}$.
\end{proof}

Lemma~\ref{lem:pph_active} means that
there exists a position $s$ which splits all positions in $\para{t}[1:n]$
into two intervals, similar to the \emph{active position} in \cite{OPH}.
Positions in $[1:s-1]$ and $[s:n]$ are called primary and secondary positions, respectively.
We also call the position $s$ as active position.

Assume we have constructed $\pph{\para{t}[1:k]}$ and we want to construct 
$\pph{\para{t}[1:k+1]}$ from $\pph{\para{t}[1:k]}$.
The primary positions $1, \dots, s-1$ in $\pph{\para{t}[1:k]}$
become primary positions also in $\pph{\para{t}[1:k+1]}$,
because $\pre{\para{t}[i:k]}=\pre{\para{t}[i:k+1]}[1:k-1+1]$ holds for $1 \leq i \leq s-1$.
Therefore, we do not need to update the primary positions.

On the other hand,
the secondary positions $s, \dots, k$ require some modifications.
When inserting a new symbol,
two cases can occur.
The first case is that
$\pre{\para{t}[i:k+1]}$ is not represented in $\pph{\para{t}[1:k]}$.
In this case, a new node $\node{\pre{\para{t}[i:k+1]}}$ is created
as a child node of $\node{\pre{\para{t}[i:k]}}$
and position $i$ becomes the primary position of the new node.
The second case is that
$\pre{\para{t}[i:k+1]}$ was already represented in $\pph{\para{t}[1:k]}$.
In this case, the secondary position $i$
that is stored in $\node{\pre{\para{t}[i:k]}}$ currently should be moved
to the child node $\node{\pre{\para{t}[i:k+1]}}$,
and position $i$ becomes the secondary position of this node.

From Lemma~\ref{lem:pph_substr}, if the node $\node{\pre{\para{t}[i:k]}}$ has an edge to
 the node $\node{\pre{\para{t}[i:k+1]}}$,
 $\node{\pre{\para{t}[i+1:k]}}$ also has an edge to
 $\node{\pre{\para{t}[i+1:k+1]}}$.
Therefore,
there exists $r$, with $1 \leq s \leq r \leq k$,
that splits the interval $[s:k]$ into two subintervals $[s:r-1]$ and $[r:k]$,
such that the node $\node{\pre{\para{t}[i:k]}}$ does not have an edge
to $\node{\pre{\para{t}[i:k+1]}}$ for $s \leq i \leq r-1$,
and does have such an edge for $r \leq i \leq k$.

The above analysis leads to the following lemma
that specifies the modifications from $\pph{\para{t}[1:k]}$ to $\pph{\para{t}[1:k+1]}$.  

\begin{figure}[t]
	\centering
	\begin{minipage}[t]{0.49\hsize}
		\centering
		\includegraphics[scale=0.6]{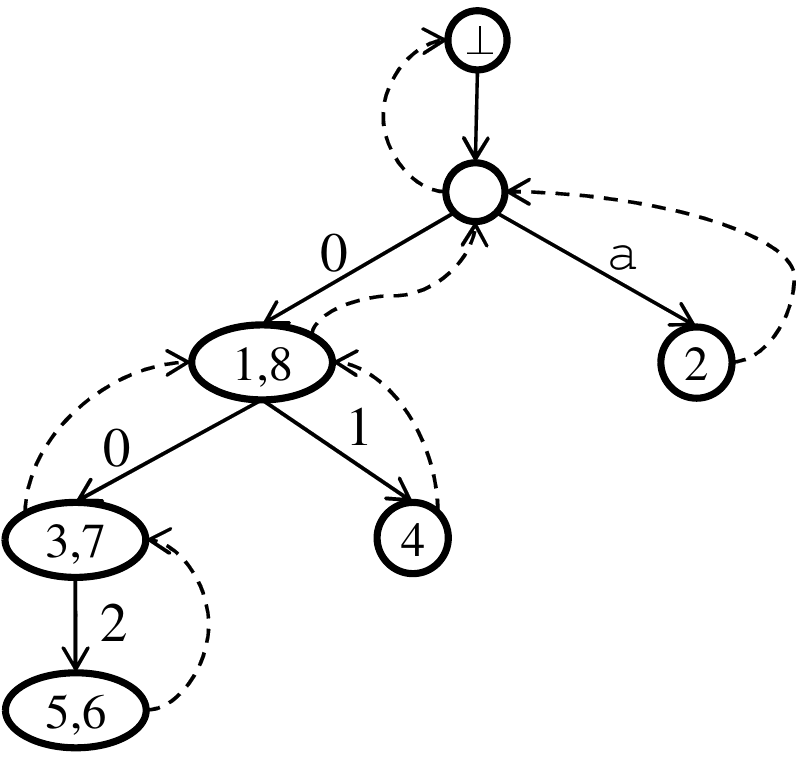}\\
		\ \ \ \scriptsize{(a)}
	\end{minipage}
	\begin{minipage}[t]{0.49\hsize}
		\centering
		\includegraphics[scale=0.6]{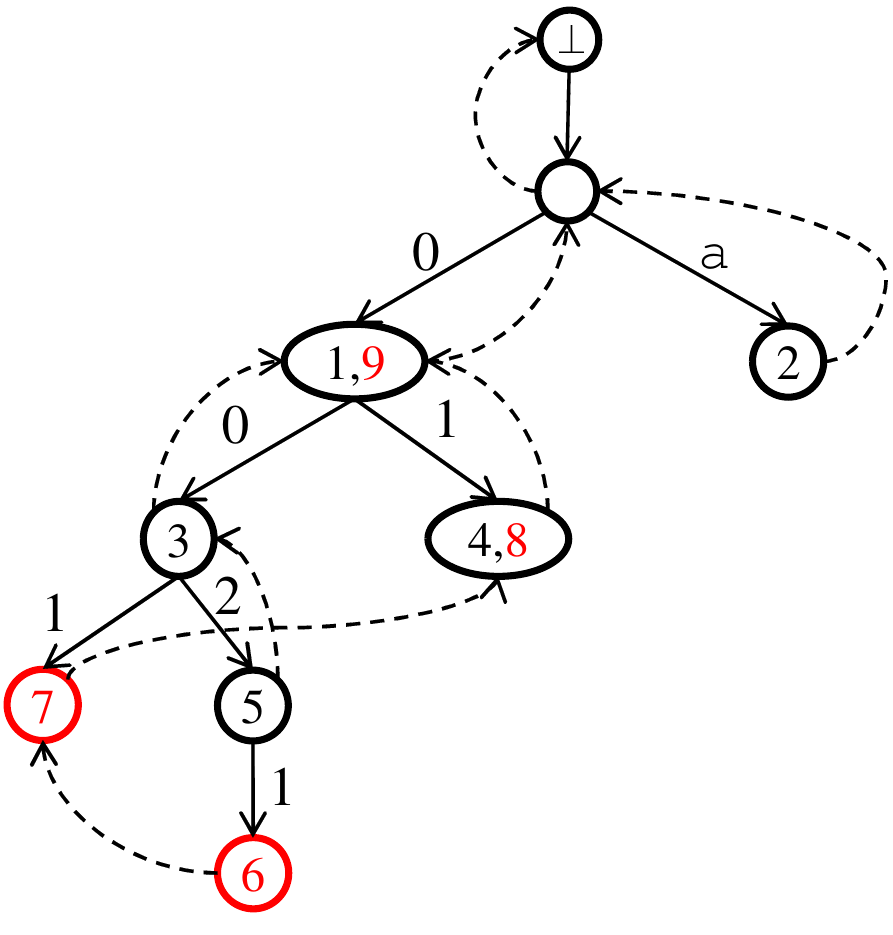}\\
		\ \ \ \scriptsize{(b)}
	\end{minipage}
	\caption{
		An example of updating a parameterized position heap, from
		(a) $\pph{x{\tt a}xyyxyx}$ to
		(b) $\pph{x{\tt a}xyyxyxx}$.
		The updated positions are colored red.
		The secondary positions $6$ and $7$ in $\pph{x{\tt a}xyyxyx}$ are become primary positions in $\pph{x{\tt a}xyyxyxx}$,
		while the secondary position $8$ in $\pph{x{\tt a}xyyxyx}$ is become a secondary position of another node in $\pph{x{\tt a}xyyxyxx}$.
		The active position is updated from $6$ to $8$.
	}
	\label{fig:pph_ex}
\end{figure}

\begin{lemma}
\label{lem:pph_modify}
Given $\para{t} \in (\Sigma \cup \Pi)^n$, consider $\pph{\para{t}[1:k]}$ for $k<n$.
Let $s$ be the active position,
stored in the node $\node{\pre{\para{t}[s:k]}}$.
Let $r \geq s$ be the smallest position such that
node $\node{\pre{\para{t}[r:k]}}$ has an outgoing edge
labeled with $\pre{\para{t}[r:k+1]}[k-r+2]$.
$\pph{\para{t}[1:k+1]}$ can be obtained by modifying
$\pph{\para{t}[1:k]}$ in the following way:
\begin{enumerate}
\item
For each node $\node{\pre{\para{t}[i:k]}}$, $s \leq i < r$,
create a new child $\node{\pre{\para{t}[i:k+1]}}$
linked by an edge labeled $\pre{\para{t}[i:k+1]}[k-i+2]$.
Delete the secondary position $i$ from the node $\node{\pre{\para{t}[i:k]}}$
and assign it as the primary position of the new node $\node{\pre{\para{t}[i:k+1]}}$,
\item
For each node $\node{\pre{\para{t}[i:k]}}$, $r \leq i \leq k$,
move the secondary position $i$ from the node $\node{\pre{\para{t}[i:k]}}$
to the node $\node{\pre{\para{t}[i:k+1]}}$.
\end{enumerate}
Moreover, $r$ will be the active position in $\pph{\para{t}[1:k+1]}$.
\end{lemma}

\begin{proof}
Consider the first case that $i$ be a secondary position in $\pph{t[1:k]}$ and $s \leq i < r$.
From the definition of $r$, there is no node $\pre{t[i:k+1]}$ in $\pph{t[i:k]}$.
Therefore, $i$ will be a primary position of the node $\pre{t[i:k+1]}$ in $\pph{t[1:k+1]}$.
We can update the position heap from $\pph{t[1:k]}$ to $\pph{t[1:k+1]}$ by delete $i$ from secondary position of the node $\pre{t[i:k]}$
and create a new node $\pre{t[i:k+1]}$ and assign $i$ to its primary position for the case $s \leq i < r$.

Next case, $i$ be a secondary position in $\pph{t[1:k]}$ and $r \leq i \le k$.
In this case, there is a node $\pre{t[i:k+1]}$ in $\pph{t[i:k]}$ and the node $\pre{t[i:k+1]}$ is also represented in $\pph{t[i:k+1]}$.
Therefore, $i$ will be a secondary position of the node $\pre{t[i:k+1]}$ in $\pph{t[1:k+1]}$.
We can update the position heap from $\pph{t[1:k]}$ to $\pph{t[1:k+1]}$ by delete $i$ from secondary position of the node $\pre{t[i:k]}$
and assign $i$ as secondary position of the node $\pre{t[i:k+1]}$ for the case $r \leq i \le k$.

Since position $i$ for $1 \leq i < r$ be a primary position in $\pph{t[1:k+1]}$ and
position $i$ for $r \leq i \leq k+1$ be a secondary position in $\pph{t[1:k+1]}$,
$r$ will be the active position in $\pph{t[1:k+1]}$.
\end{proof}

Fig.~\ref{fig:pph_ex} show an example of updating a parameterized position heap.
The modifications specified by Lemma~\ref{lem:pph_modify} need to be
applied to all secondary positions.
In order to perform these modifications efficiently,
we use parameterized suffix pointers.

\begin{definition}[Parameterized Suffix Pointer]
For each node $\node{\pre{\para{t}[i:j]}}$ of $\pph{\para{t}}$,
the \emph{parameterized suffix pointer} of $\node{\pre{\para{t}[i:j]}}$ is defined by
$\psufp{\node{\pre{\para{t}[i:j]}}}=\node{\pre{\para{t}[i+1:j]}}$.
\end{definition}
By Lemma~\ref{lem:pph_substr},
whenever the node $\node{\pre{\para{t}[i:j]}}$ exists,
the node $\node{\pre{\para{t}[i+1:j]}}$ exists too.
This means that $\psufp{\node{\pre{\para{t}[i:j]}}}$
always exists.
During the construction of the parameterized position heap,
let $\bot$ be the auxiliary node that
works as the parent of $\rt$
and is connected to $\rt$ with an edge
labeled with any symbol $c \in \Sigma \cup {0}$.
We define $\psufp{\rt}=\bot$.

When $s$ is the active position in $\pph{\para{t}[1:k]}$,
we call $\node{\pre{\para{t}[s:k]}}$ the \emph{active node}.
If no node holds a secondary position,
$\rt$ becomes the active node and the active position is set to $k+1$.
The nodes for the secondary positions $s$, $s+1$, $\dots$, $k$
can be visited by traversing with the suffix pointers from the active node.
Thus, the algorithm only has to memorize the active position and the active node
in order to visit any other secondary positions.

Updating $\pph{\para{t}[1:k]}$ to $\pph{\para{t}[1:k+1]}$
specified by Lemma~\ref{lem:pph_modify} is processed as the following procedures.
The algorithm traverses with the suffix pointers
from the active node till the node that has the outgoing edge labeled with $\pre{\para{t}[i:k+1]}[k-i+2]$
is found, which is $i=r$.
For each traversed node,
a new node is created and linked by an edge labeled with $\pre{\para{t}[i:k+1]}[k-i+2]$
to each node.
A suffix pointer to this new node is set from the previously created node.
When the node that has the outgoing edge labeled with $\pre{\para{t}[i:k+1]}[k-i+2]$
is traversed,
the algorithm moves to the node that is led to by this edge,
and a suffix pointer to this node is set from the last created node,
then the algorithm assigns this node to be the active node.

\begin{algorithm2e}[!t]
\caption{Parameterized position heap online construction algorithm}
\label{alg:pph_algo}
\KwIn{A p-string $\para{t} \in (\Sigma \cup \Pi)^n$ }
\KwOut{A parameterized position heap $\pph{\para{t}}$}
\SetKwData{CURRENT}{currentNode}
\SetKwData{NEXT}{nextNode}
\SetKwData{ROOT}{root}
\SetKwData{UNDEFINED}{undefined}
\SetKwData{LAST}{lastCreateNode}
$\create$ $\rt$ and $\perp$ nodes\;
$\psufp{\rt} =\ \perp$\;
$\child{\perp}{c} = \rt$ for $c \in \Sigma \cup \{0\}$\;
$\CURRENT = \rt$\;
$s = 1$\;
\For{$i=1$ \textsf{to} $n$}{
	$c = \norm{\pre{t}[i]}{\depth{\CURRENT}}$\;
	$\LAST = \UNDEFINED$\;
	\While{$\child{\CURRENT}{c} = \nll$}{
		create $\newactn$\;
		$\primary{\newactn} = s$\;
		$\child{\CURRENT}{c}=\newactn$\;
		\lIf{$\LAST \neq \UNDEFINED$}{
			$\psufp{\LAST} = \newactn$}
		$\LAST = \newactn$\;
		$\CURRENT = \psufp{\CURRENT}$\;
		$c = \norm{\pre{t}[i]}{\depth{\CURRENT}}$\;
		$s = s+1$\;
	}
	$\CURRENT = \child{\CURRENT}{c}$\;
	\lIf{$\LAST \neq \UNDEFINED$}{
		$\psufp{\LAST} = \CURRENT$}
}
\While{$s \leq n$}{
	$\secondary{\CURRENT} = s$\;
	$\CURRENT = \psufp{\CURRENT}$\;
	$s = s+1$\;
}
\end{algorithm2e}

A pseudocode of our proposed construction algorithm
is given as Algorithm~\ref{alg:pph_algo}.
$\primary{v}$ and $\secondary{v}$ denotes primary and secondary positions of $v$, respectively.
From the property of prev-encoding, $\pre{\para{t}[i+1:k+1]}[k-i+1] = \pre{\para{t}[i:k+1]}[k-i+2]$ if $\pre{\para{t}[i:k+1]}[k-i+2] \in \Sigma$
or $\pre{\para{t}[i:k+1]}[k-i+2] \leq k-i$
and $\pre{\para{t}[i+1:k]}[k-i+1] = 0$ otherwise.
Therefore, we use a function $\norm{c}{j}$ that returns $c$ if $c\in\Sigma$ or $c \leq j$ and returns $0$ otherwise.

The construction algorithm consists of $n$ iterations.
In the $i$-th iteration,
the algorithm read $\para{t}[i]$ and make $\pph{\para{t}[1:i]}$.
In the $i$-th iteration,
the traversal of the suffix pointers as explained above is done.
Since the depth of the current node decreases by traversing a suffix pointer,
the number of the nodes that can be visited by traversal is $O(n)$.
For each traversed node,
all the operations such as creating a node, an edge and updating position can be done in $O(\log{(|\Sigma|+|\Pi|)})$.
Therefore,
the total time for the traversals is $O(n\log{(|\Sigma|+|\Pi|)})$.

From the above discussion,
the following theorem
is obtained.

\begin{theorem}
Given $\para{t} \in (\Sigma \cup \Pi)^n$,
Algorithm~\ref{alg:pph_algo} constructs $\pph{\para{t}}$
in $O(n\log{(|\Sigma|+|\Pi|)})$ time and space.
\end{theorem}

\subsection{Augmented Parameterized Position Heaps}
We will describe \emph{augmented parameterized position heaps},
the parameterized position heaps with an additional data structure
called the \emph{parameterized maximal-reach pointers}
similar to the maximal-reach pointers for the position heap~\cite{PH}.
The augmented parameterized position heap
gives an efficient algorithm for parameterized pattern matching.

\begin{definition}[Parameterized Maximal-Reach Pointer]
For a position $i$ on $\para{t}$, 
a \emph{parameterized maximal-reach pointer of $\pmrp{i}$} is a pointer from node $i$
to the deepest node whose path label is a prefix of $\pre{\para{t}[i:]}$.
\end{definition}

Obviously, if $i$ is a secondary position, then $\pmrp{i}$ is node $i$ itself.
We assume that the parameterized maximal-reach pointer for a double node
applies to the primary position of this node.
Fig.~\ref{fig:pph} (b) shows an example of an augmented parameterized position heap.
Given a prev-encoded p-string $\pre{\para{w}}$
represented in an augmented parameterized position heap $\apph{\para{t}}$ and a position $1 \leq i \leq n$,
we can determine whether $\pre{\para{w}}$ is a prefix of $\pre{\para{t}[i:]}$ or not in $O(1)$ time
by checking whether $\pmrp{i}$ is a descendant of $\pre{\para{w}}$ or not.
It can be done in $O(1)$ time
by appropriately preprocessing $\apph{t}$~\cite{INTROALG}.

Parameterized maximal-reach pointers can be computed by using parameterized suffix pointers, similar to \cite{OPH}.
Algorithm~\ref{alg:apph_algo} shows an algorithm to compute parameterized maximal-reach pointers.
$\pmrp{i}$ is computed iteratively for $i= 1, 2, \cdots, n$.
Assume that we have computed $\pmrp{i}$ for some $i$.
Let $\pmrp{i}=\pre{\para{t}[i:l]}$.
Obviously, $\pre{\para{t}[i+1:l]}$ is a prefix of the string
represented by $\pmrp{i+1}$.
Thus, in order to compute $\pmrp{i+1}$,
we should extend the prefix $\pre{\para{t}[i+1:l]} =\psufp{\pre{\para{t}[i:l]}}$ in $\pph{t}$
until we found $l'$ such that node $\pre{\para{t}[i+1:l']}$ does not have outgoing edge labeled with $\pre{\para{t}[i+1:]}[l'-i+1]$
and set $\pmrp{i+1} = \pre{\para{t}[i+1:l']}$.
In this time,
we need re-compute $\pre{\para{t}[i+1:]}$
by replacing $\pre{\para{t}[i+1:]}[j]$ with $0$ if we found that $\pre{\para{t}[i+1:]}[j] \geq j$.
The total number of extending $\pre{\para{t}[i+1:l]}$ in the algorithm
is at most $n$ because both $i$ and $l$ always increase in each iteration.
In each iteration, operations such as traversing a child node 
can be done in $O(\log{(|\Sigma|+|\Pi|)})$.
Therefore, we can get the following theorem.
\begin{theorem}
	Parameterized maximal-reach pointers for $\pph{\para{t}}$ 
	can be computed in $O(n\log{(|\Sigma|+|\Pi|)})$ time.
\end{theorem}

\begin{algorithm2e}[!t]
	\caption{Augmented parameterized position heap construction algorithm}
	\label{alg:apph_algo}
	\KwIn{A p-string $\para{t} \in (\Sigma \cup \Pi)^n$ and $\pph{t}$}
	\KwOut{An augmented parameterized position heap $\apph{\para{t}}$}
	\SetKwData{CURRENT}{currentNode}
	\SetKwData{NEXT}{nextNode}
	\SetKwData{ROOT}{root}
	\SetKwData{UNDEFINED}{undefined}
	\SetKwData{LAST}{lastCreateNode}
	let $t[n+1] = \$$ where \$ is a symbol that does not appear in $t$ elsewhere\;
	$\CURRENT = \rt$\;
	$l = 1$\;
	\For{$i=1$ \textsf{to} $n$}{
		$c = \norm{\pre{t}[l]}{l-i}$\;
		\While{$\child{\CURRENT}{c} \neq \nll$}{
			$\CURRENT = \child{\CURRENT}{c}$\;
			$l = l + 1$\;
			$c = \norm{\pre{t}[l]}{l-i}$\;
		}
		$\pmrp{i} = \CURRENT$\;
		$\CURRENT = \psufp{\CURRENT}$\;
	}
\end{algorithm2e}

\subsection{Parameterized Pattern Matching with Augmented Parameterized Position Heaps}

\begin{algorithm2e}[!t]
\caption{Parameterized pattern matching algorithm with APPH}
\label{alg:PARA_MRP_MATCH}
\KwIn{$\para{t} \in (\Sigma \cup \Pi)^n$ , $\para{p} \in (\Sigma \cup \Pi)^m$, and  $\apph{\para{t}}$}
\KwOut{The list $ans$ of position $i$ such that $\pre{\para{p}}=\pre{\para{t}[i:i+m-1]}$ }
\SetKwData{BREAK}{break}
\SetKwData{EMPTYLIST}{empty list}

let $w$ be the longest prefix of $\pre{p}$ represented in $\apph{t}$ and $u$ be the node represents $w$\;
\If{$|w|=m$}{
	$v = \rt$\;
	\For{$i=1$ {\bf to} $m$}{
	\label{ALG_PMM:VISIT}
		$v$ = $\child{v}{\pre{p}[i]}$\;
		\lIf{$\pmrp{v} \in \descendant{\apph{\para{t}}}{u}$}{
			add $\primary{v}$ to $\ans$}
	}
	add all primary and sedondary position of decendants of $u$ to $ans$\;\label{ALG_PMM:TRAV1}
}
\Else{
	$v = \rt$\;
	$i=1,j=1$\;
	\While{$i \leq |w|$}{
		$v$ = $\child{v}{\pre{p}[i]}$\;
		$i = i+1$\;
		\lIf{$\pmrp{v} = u$}{
			add $\primary{v}$ to $\ans$}
	}
	\While{$i \neq m$}{		\label{ALG_PMM:LOOP1}
	$j = i, v = \rt$\;
	$\zero$ = \EMPTYLIST\;
	\While{$i \neq m$}{ \label{ALG_PMM:LOOP2}
		$c = \norm{\pre{p}[i]}{i-j}$\;
		\lIf{$\child{v}{c} =\nll$}{
			{\bf break}}
		\lIf{$c = 0$}{
			add $i$ to $\zero$}
		$v = \child{v}{c}$\;
		$i = i+1$\;
	}
		\lIf{$v=\rt$}{
			{\bf return} \EMPTYLIST}
		\ForEach{$i' \in \ans$}{	\label{ALG_PMM:LOOP3}
			\If{$i=m$}{ \label{ALG_PMM:VERIFY_FIRSTOCC}
				\lIf{$\pmrp{i'+j-1} \notin \descendant{\apph{\para{t}}}{v}$}{
					remove $i'$ from $ans$}	\label{ALG_PMM:VERIFY_OCC1}
			}
			\Else{
				\lIf{$\pmrp{i'+j-1} \neq v$}{
					remove $i'$ from $ans$}	\label{ALG_PMM:VERIFY_OCC2}
			}
			\For{$k=1$ {\bf to} $|\zero|$}{
				\If{$\norm{\pre{t}[i'+\zero[k]-1]}{\zero[k]-1} \neq \pre{p}[\zero[k]]$}{
					remove $i'$ from $ans$;\ \ \ 
				}
			}
		}	\label{ALG_PMM:LOOP3_END}
	} 
}
{\bf return} $\ans$\;
\end{algorithm2e}

Ehrenfeucht \etal{}~\cite{PH} and Kucherov~\cite{OPH} split a pattern $p$ into segments
$q_1, q_2, \cdots, q_k$,
then compute occurrences of $q_1 q_2 \cdots q_{j}$ iteratively for $j=1, \cdots, k$.
The correctness depends on a simple fact that for strings
$x = t[i:i+|x|-1]$ and $y = t[i+|x|:i+|x|+|y|-1]$ implies $xy = t[i:i+|xy|-1]$.
However, when $x$, $y$, and $t$ are p-strings, $\pre{\para{x}} = \pre{\para{t}[i:i+|\para{x}|-1]}$ and
$\pre{\para{y}} = \pre{\para{t}[i+|\para{x}|:i+|\para{x}|+|\para{y}|-1]}$ does not necessarily implies 
$\pre{\para{xy}} = \pre{\para{t}[i:i+|\para{xy}|-1]}$.
Therefore, we need to modify the matching algorithm for parameterized strings.

Let $\para{x}$, $\para{y}$ and $\para{w}$ be p-strings such that
$|\para{w}|=|\para{xy}|$, $\pre{\para{x}}=\pre{\para{w}[:|x|]}$ and
$\pre{\para{y}}=\pre{\para{w}[|x|+1:]}$.
Let us consider the case that $\pre{\para{xy}} \neq \pre{\para{w}}$.
From $\pre{\para{x}}=\pre{\para{w}[:|x|]}$ and $\pre{\para{y}}=\pre{\para{w}[|x|+1:]}$,
$\para{x}$ and $\para{y}$ have the same structure of $\para{w}[:|x|]$ and $\para{w}[|x|+1:]$,
respectively.
However, the parameter symbols those are prev-encoded into $0$ in $\pre{\para{y}}$ and $\pre{\para{w}[|x|+1:]}$,
might be encoded differently in $\pre{xy}$ and $\pre{w}$, respectively.
Therefore, we need to check whether $\pre{xy}[|x|+i] = \pre{w}[|x|+i]$ if $\pre{y}[i]=0$.
Given $\pre{\para{xy}}$ and the set of positions of $0$ in $\pre{\para{y}}$,
$\zero = \{ i~|~1 \leq i \leq |\para{y}|~{\rm such~that}~\pre{\para{y}}[i]=0 \}$.
We need to verify whether $\pre{xy}[|x|+i] = \pre{w}[|x|+i]$ or not for $i \in \zero$.
Since the size of $\zero$ is at most $|\Pi|$,
this computation can be done in $O(|\Pi|)$ time.

\begin{figure}[t]
	\centering
	\begin{minipage}[t]{0.32\hsize}
		\centering
		\includegraphics[scale=0.45]{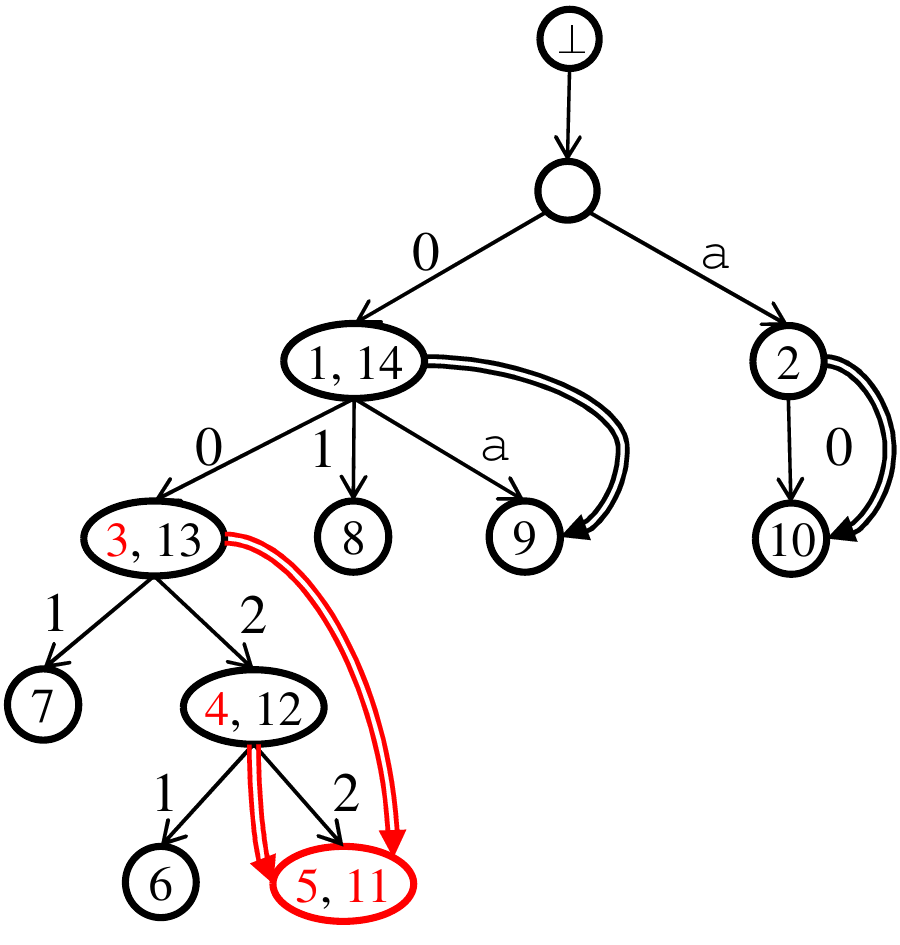}\\
		\ \ \ \scriptsize{(a)}
	\end{minipage}
	\begin{minipage}[t]{0.66\hsize}
		\centering
		\includegraphics[scale=0.45]{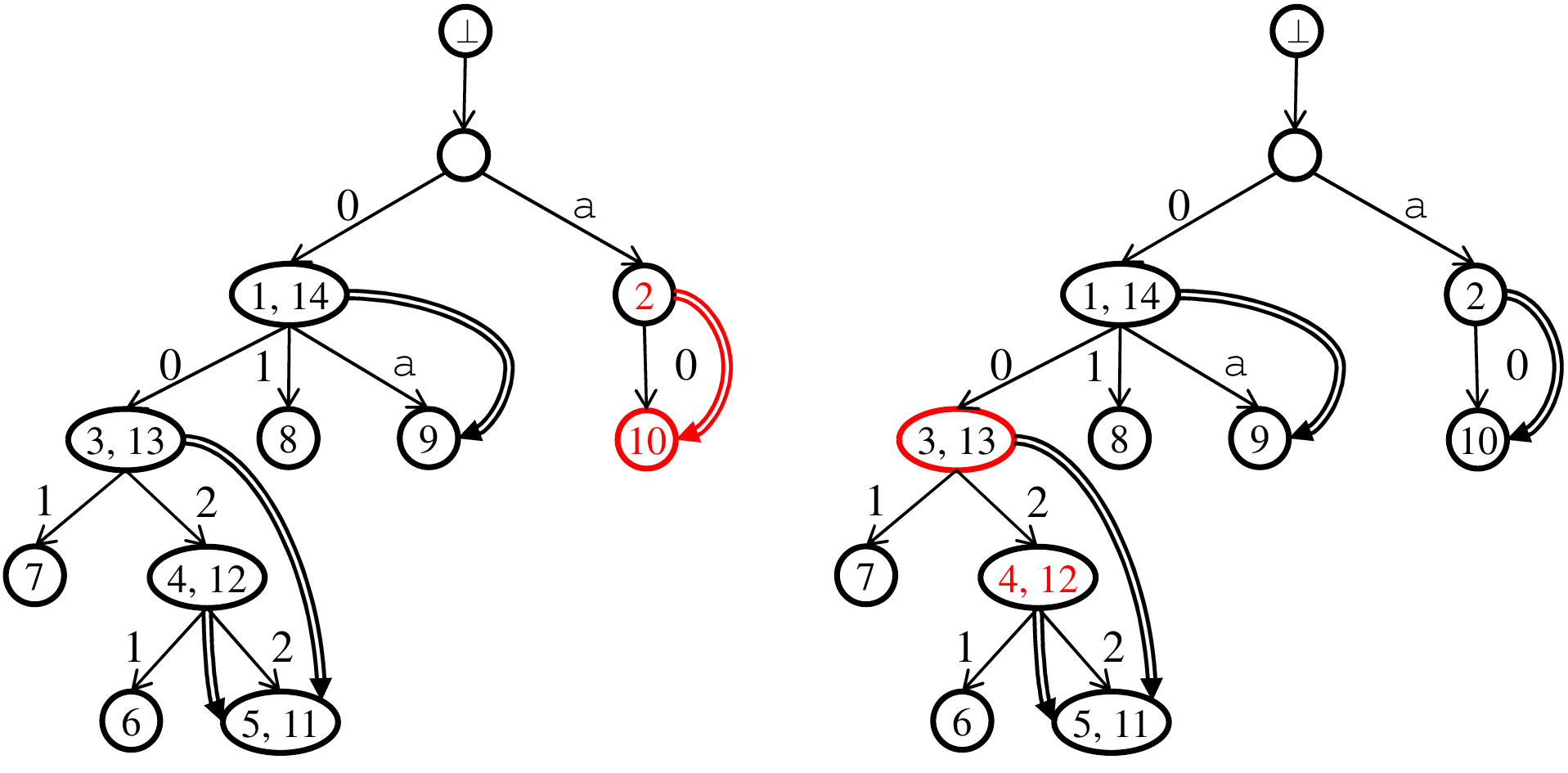}\\
		\scriptsize{(b)}
	\end{minipage}
	\caption{
		Examples of finding occurrence positions of a pattern using an augmented parameterized position heap $\pph{x{\tt a}xyxyxyy{\tt a}xyxy}$.
		(a) Finding $xyxy$ ($\pre{xyxy} = 0022$).
		(b) Finding ${\tt a}xyx$ ($\pre{{\tt a}xyx} = {\tt a}002$).
	}
	\label{fig:apph_match}
\end{figure}

A pseudocode of proposed matching algorithm for the parameterized pattern matching problem
is shown in Algorithm~\ref{alg:PARA_MRP_MATCH}.
$\descendant{\apph{\para{t}}}{u}$ denotes the set of all descendants of node $u$ in $\apph{t}$ including node $u$ itself.
The occurrences of $p$ in $t$ have the following properties on $\apph{t}$.
\begin{lemma} \label{lem:occur}
	If $\pre{p}$ is represented in $\apph{t}$ as a node $u$ then $p$ occurs at position $i$ iff $\pmrp{i}$ is $u$ or its descendant.
\end{lemma}
\begin{proof}
	Let $u$ be a node represents $\pre{p}$. Assume $p$ occurs at position $i$ in $t$ and represented in $\apph{t}$ as $\pre{t[i:k]}$.
	Since either $\pre{t[i:k]}$ is a prefix of $\pre{p}$ or $\pre{p}$ is a prefix of $\pre{t[i:k]}$,
	then $i$ is either an ancestor or descendant of $u$.
	For both cases $\pmrp{i}$ is a descendant of $u$,
	because $p$ occurs at position $i$.  
	
	Next let $i$ be a node such that $\pmrp{i}$ is a descendant of $u$ and represents $\pre{t[i:k]}$.
	In this case, $\pre{p}$ is a prefix of $\pre{t[i:k]}$.
	Therefore $p$ occurs at $i$.
\end{proof}

\begin{lemma} \label{lem:split}
	Assume $\pre{p}$ is not represented in $\apph{t}$.
	We can split $p$ into $q_1, q_2, \cdots, q_k$ such that $q_j$ is the longest prefix of $\pre{p[|q_1\cdots q_{j-1}|+1:]}$ that is represented in $\apph{t}$.
	If $p$ occurs at position $i$ in $t$,
	then $\pmrp{i+|q_1\cdots q_{j-1}|}$
	is the node $\pre{q_j}$ for $1 \leq j < k$
	and $\pmrp{i+|q_1\cdots q_{k-1}|}$ is the node $\pre{q_k}$ or its descendant.
\end{lemma}
\begin{proof}
	Let $p=q_1 q_2 \cdots q_k$ occurs at position $i$ in $t$.
	Since $\pre{q_1}$ is a prefix of $\pre{p}$,
	then $\pmrp{i}$ is the node that represents $\pre{q_1}$ or its descendant.
	However, if $\pmrp{i}$ is a descendant of node $\pre{q_1}$,
	then we can extend $q_1$ which contradicts with the definition of $q_1$.
	Therefore, $\pmrp{i}$ is the node represents $\pre{q_1}$.
	
	Similarly for $1<j<k$, $\pre{q_j}$ is a prefix of $\pre{p[|q_1\cdots q_{j-1}|+1:]}$ and occurs at position $i + |q_1\cdots q_{j-1}|$ in $t$.
	Therefore, $\pmrp{i + |q_1\cdots q_{j-1}|}$ is the node represents $\pre{q_j}$.
	Last, since $q_k$ is a suffix of $p$, then $\pmrp{i + |q_1\cdots q_{j-1}|}$ can be the node $\pre{q_k}$ or its descendant.
\end{proof}

Algorithm~\ref{alg:PARA_MRP_MATCH} utilizes Lemmas~\ref{lem:occur}~and~\ref{lem:split} to find occurrences of $p$ in $t$ by using $\apph{t}$.
First, if $\pre{p}$ is represented in $\apph{t}$ then the algorithm will output all position $i$ such that $\pmrp{i}$ is a node $\pre{p}$ or its descendant.
Otherwise, it will split $p$ into $q_1 q_2 \cdots q_k$ and find their occurrences as described in Lemma~\ref{lem:split}.
The algorithm also checks whether $\pre{q_1 \cdots q_j}$ occurs in $t$ or not in each iteration as described the above.

Examples of parameterized pattern matching by using an augmented position heap are given in Fig.~\ref{fig:apph_match}.
Let $t = x{\tt a}xyxyxyy{\tt a}xyxy$ be the text.
In Fig.~\ref{fig:apph_match} (a) we want to find the occurrence positions of a pattern $p_1 = xyxy$ in $t$.
In this case, since $\pre{p_1} = 0022$ is represented in $\pph{t}$,
The algorithm outputs all positions $i$ such that $\pmrp{i}$ is the node $0022$ or its descendants, those are $3$, $4$, $5$, and $11$.
On the other hand, Fig.~\ref{fig:apph_match} (b) shows how to find the occurrence positions of a pattern $p_2 = {\tt a}xyx$ in $t$.
In this case, $\pre{p_2} = {\tt a}002$ is not represented in $\pph{t}$.
Therefore, The algorithm finds the longest prefix of $\pre{p_2}$ that is represented in $\pph{i}$, which is $\pre{p_2}[1:2] ={\tt a}0$.
We can see that $prmp(2) = pmrp(10) = {\tt a}0$, then we save positions $2$ and $10$ as candidates to $\ans$.
Next, The algorithm finds the node that represents the longest prefix of $\pre{p_2[3:]} = 00$ which is $\pre{p_2[3:]} = 00$ itself.
Since both of $\pmrp{2+|p_2[1:2]|} = \pmrp{4}$ and $\pmrp{10 + |p_2[1:2]}| = \pmrp{12}$
is descendants of the node $00$, $\pre{t[2:5][3]} = \pre{t[10:13][3]} = \pre{p_2}[[3]] = 0$,
and $\pre{t[2:5][4]} = \pre{t[10:13][4]} = \pre{p_2}[4] = 2$,
then the algorithm outputs $2$ and $10$.

The time complexity of the matching algorithm is as follow.
\begin{theorem}
Algorithm~\ref{alg:PARA_MRP_MATCH} runs in $O(m\log{(|\Sigma|+|\Pi|)}+m|\Pi|+\occ)$ time.
\end{theorem}
\begin{proof}
It is easily seen that we can compute line~\ref{ALG_PMM:VISIT} to \ref{ALG_PMM:TRAV1}
in $O(m\log{(|\Sigma|+|\Pi|)}+\occ)$ time.
Assume that
$\para{p}$ can be decomposed into $\para{q}_{1}$, $\para{q}_{2}$, $\cdots$, $\para{q}_{k}$
such that $\para{q}_{1}$ is the longest prefix of $\para{p}$ and $\para{q}_{i}$ is the longest prefix of $\pre{\para{p}[|q_1\cdots q_{j-1}|+1:]}$ represented in $\apph{t}$.
The loop for line~\ref{ALG_PMM:LOOP1} consists of $k-1$ iterations.
In the loop line~\ref{ALG_PMM:LOOP2} in $j$-th iteration,
$\para{q}_{j+1}$ is extended up to reach $|\para{q}_{j+1}|$ length.
This can be computed in $O(|\para{q}_{j+1}|\log{(|\Sigma|+|\Pi|)})$ time.
After $k-1$ iterations, the total number of extending of $\para{q}_{j+1}$
does not exceed $m$, because $\Sigma_{j=2}^k |\para{q_j}| < m$.
In the loop for line~\ref{ALG_PMM:LOOP3},
the algorithm verifies elements of $ans$.
In $j$-th iteration, the size of $ans$ is at most $|\para{q}_{j}|$.
Thus, after $k-1$ iterations, the total number of elements
verified in line~\ref{ALG_PMM:LOOP3} does not exceed $m$
by the same reason for that of line~\ref{ALG_PMM:LOOP2}.
In each verification in line~\ref{ALG_PMM:LOOP3},
the number of checks for line~\ref{ALG_PMM:VERIFY_OCC1} and~\ref{ALG_PMM:VERIFY_OCC2} is at most $|q_j|$.
Therefore,
it can be computed from line~\ref{ALG_PMM:LOOP3} to \ref{ALG_PMM:LOOP3_END}
in $O(m|\Pi|)$ time.
\end{proof}

\section{Conclusion and Future Work}
For the parameterized pattern matching problem,
we proposed an indexing structure called a parameterized position heap.
Given a p-string $\para{t}$ of length $n$ over a constant size alphabet,
the parameterized position heap for $\para{t}$
can be constructed in $O(n\log{(|\Sigma|+|\Pi|)})$ time
by our construction algorithm.
We also proposed an algorithm for the parameterized pattern matching problem.
It can be computed in $O(m\log{(|\Sigma|+|\Pi|)} + m|\Pi| +occ)$ time
using parameterized position heaps
with parameterized maximal-reach pointers.
Gagie \etal{}~\cite{Gagie2013} showed an interesting relationship between position heap and suffix array of a string.
We will examine this relation for parameterized position heap and parameterized suffix array~\cite{PSC2008-8,I2009} as a future work.


\bibliographystyle{abbrv} 
\bibliography{docs/ref}

\end{document}